\documentclass{comjnl}
\usepackage{amssymb}
\usepackage[cmex10]{amsmath}
\usepackage{acronym}
\usepackage{algorithmic}
\allowdisplaybreaks[4]
\usepackage{graphicx}
\usepackage{color}
\usepackage{aecompl}
\begin{document}
\title{Hitting Times of Random Walks on Edge Corona Product Graphs}
\author{Mingzhe~Zhu} %\email{17307130099@fudan.edu.cn}
%\affiliation{Shanghai Key Laboratory of Intelligent Information
%	Processing \& Shanghai Engineering Research Institute of Blockchain, School of Computer Science, Fudan University, Shanghai 200433, China}
\affiliation{School of Computer Science, Fudan University, Shanghai 200433, China}
\affiliation{Shanghai Key Laboratory of Intelligent Information Processing,   Fudan University, Shanghai 200433, China}

\author{Wanyue~Xu} %\email{xuwy@fudan.edu.cn}
\affiliation{School of Computer Science, Fudan University, Shanghai 200433, China}
\affiliation{Shanghai Key Laboratory of Intelligent Information Processing,   Fudan University, Shanghai 200433, China}

\author{Wei~Li} \email{fd\_liwei@fudan.edu.cn}

\affiliation{Academy for Engineering and Technology, Fudan University, Shanghai 200433, China}

\author{Zhongzhi~Zhang} \email{fd\_liwei@fudan.edu.cn; zhangzz@fudan.edu.cn}
\affiliation{School of Computer Science, Fudan University, Shanghai 200433, China}
\affiliation{Shanghai Key Laboratory of Intelligent Information Processing,   Fudan University, Shanghai 200433, China}
\affiliation{Institute of Intelligent Complex Systems, Fudan University, Shanghai 200433, China}
\affiliation{Shanghai Engineering Research Institute of Blockchain,  Fudan University, Shanghai 200433, China}

\author{Haibin~Kan} %\email{hbkan@fudan.edu.cn}
\affiliation{School of Computer Science, Fudan University, Shanghai 200433, China}
\affiliation{Shanghai Key Laboratory of Intelligent Information Processing,   Fudan University, Shanghai 200433, China}
\affiliation{Shanghai Engineering Research Institute of Blockchain,  Fudan University, Shanghai 200433, China}
\affiliation{Yiwu Research Institute of Fudan University, Yiwu City 322000, China}

\shortauthors{M. Zhu, W. Xu, W. Li, Z. Zhang and H. Kan}

\begin{abstract}
  Graph products have been extensively applied to model complex networks with striking properties observed in real-world complex systems. In this paper, we study the hitting times for random walks on a class of graphs generated iteratively by edge corona product.  We first derive recursive solutions to the eigenvalues and eigenvectors of the normalized adjacency matrix associated with the graphs. Based on these results, we further obtain interesting quantities about hitting times of random walks, providing iterative formulas for two-node hitting time, as well as closed-form expressions for the Kemeny's constant defined as a weighted average of hitting times over all node pairs, as well as the  arithmetic mean of hitting times of all pairs of nodes.
\end{abstract}

\maketitle

\section{Introduction}
Graph operations and products play an important role in network science, which have been used to model complex networks with the prominent scale-free~\cite{BaAl99} and small-world~\cite{WaSt98} properties as observed in various real-life networks~\cite{Ne03}. Since diverse realistic large-scale networks consist of smaller pieces or patterns, such as communities~\cite{GiNe02}, motifs~\cite{MiShItKaKhAl02}, and cliques~\cite{Ts15}, graph operations and products are a natural way to generate a massive graph out of smaller ones. Furthermore, there are many advantages to using graph operations and products to create complex networks. For example, it allows  analytical treatment for structural and dynamical aspects of the resulting networks. Thus far, a variety of graph operations and products have been introduced or proposed to construct models of complex networks, including triangulation~\cite{DoGoMe02,ZhRoZh07,XiZhCo16b,ShZh19,YiZhPa20}, Kronecker product~\cite{We62,LeFa07,LeChKlFaGh10}, hierarchical product~\cite{BaCoDaFi09,BaDaFiMi09,BaCoDaFi16,QiYiZh19,QiZhYiLi19}, as well as corona product~\cite{LvYiZh15,ShAdMi17,QiLiZh18}.

%Although various models using graph operations or products have been developed, most of existing works only capture the pairwise relation among system elements, neglecting those higher-order interactions among nodes, in spite of the fact that these group or simplicial interactions are ubiquitous in various real systems~\cite{BeGlLe16,BeAbScJaKl18,SaCaDaLa18}. For example, in a group of online social media~\cite{Le20}, such as WhatsApp, Telegram, and WeChat, information diffusion is a one-to-many process. In such case, interactions among members of a group is higher-order. Again for instance, in a scientific collaboration network~\cite{PaPeVa17}, for a paper with more than two authors, the interactions among the authors are not pairwise, but instead higher-order. In addition to social media or networks, simplicial interactions also appear in neuronal spiking activities~\cite{GiPaCuIt15,ReNoScect17}, proteins~\cite{WuOtBa03}, among other complex systems. These higher-order interactions have profound influences on various dynamical processes, such as percolation~\cite{BiZi18}, synchronization~\cite{SkAr19}, and epidemic spreading~\cite{MaGoAr20}.

%The higher-order interactions observed in complex systems can be modeled via several mathematical tools, for example, simplicial complexes~\cite{CoBi17,PeBa18, LaPeBaLa19} and hypergraphs~\cite{ Be84,GhZlCaNe09}, both of which have become popular models for complex systems involving higher-order organization.

Recently, a class of  iteratively growing network model was introduced, leveraging  an edge operation on graphs~\cite{WaYiXuZh19}. This family of graphs  exhibit the striking scale-free small-world properties as observed in diverse real systems. The degree distribution  $P(d)$ of the graphs follows a power-law form $P(d)\sim d^{-\gamma_q}$ with the exponent  $\gamma_q$ lying in the interval $(2,3)$. Their diameter scales logarithmically with the number of nodes. Moreover, their clustering coefficient is high. However, except some structural and combinatorial properties, the dynamical aspects on these networks are not well understood, for example, hitting times of random walks on this network family. 
 
 % although it is expected that understanding the dynamics on this network family is instrumental to find the role of higher-order structure in dynamical processes.

In this paper, we present an in-depth study on hitting time--- a most relevant quantity about random walks on the iteratively growing networks~\cite{WaYiXuZh19}. We first give iterative formulas for eigenvalues and eigenvectors of normalized adjacency (or Laplacian) matrix for the networks, based on which we determine two-node hitting time and the  Kemeny's constant for random walks. Also, we derive closed-form expressions for the sum of hitting times, additive-degree sum of hitting times, multiplicative-degree sum of hitting times over all pairs of nodes, as well as the arithmetic mean of hitting times of all node pairs.

\section{Preliminaries}

In this section, we introduce some basic concepts for graphs and random walks on graphs.

\subsection{Graph and Matrix Notation}

Let $G(V,E)$ denote a simple connected graph with $n$ nodes/vertices and $m$ edges. Let $V(G)=\{1,2,\ldots,n\}$ be the  set of  $n$ nodes, and let $E(G)=\{e_1,e_2,\ldots, e_m\}$ be  set of $m$ edges.

Let $A$ denote the adjacency matrix of $G$, the $(i,j)$th entry $A(i,j)$ of which is 1 (or 0) if nodes $i$ and $j$ are (not) adjacent in $G$. Let $\Psi (i)$ denote the set of neighbors for  node $i$  in graph $G$. Then the degree of node $i$ is  $d_i=\sum_{j  \in \Phi(i)}A(i,j)$, which forms the $i$th diagonal entry of the diagonal degree matrix $D$ for $G$. The incidence matrix $B$ of  $G$ is an $n\times m$, where the $(i,j)$th entry  $B(i,j)=1$ (or 0) if node $i$ is (not) incident with $e_j$.
\begin{lemma}\emph{\cite{CvDoSa80}} \label{bRank}
	Let $G$ be a simple connected unbipartite graph with $n$ nodes. Then the rank of its incidence matrix $B$ is ${\rm rank} (B)=n$.
\end{lemma}
\begin{lemma}\emph{\cite{CvDoSa80}} \label{BBAD}
	Let $G$ be a simple connected graph. Then its incidence matrix $B$, adjacency matrix $A$ and diagonal degree matrix $D$ satisfy
	\begin{equation*}
		BB^\top = A + D.
	\end{equation*}
\end{lemma}
\subsection{Random Walks on Graphs}

For a graph $G$, one can define a discrete-time unbiased random walk running on it. At every time step, the walker jumps from its current location, node $i$,  to an adjacent node $j$ with probability $A(i,j)/d_i$. Such a random walk on $G$ is  a Markov chain~\cite{KeSe60} characterized by the transition probability matrix $T=D^{-1}A$, with its  $(i,j)$th entry  $T(i,j)$ being $A(i,j)/d_i$. For an unbiased random walk on unbipartite graph $G$ with $n$ nodes and $m$ edges, its stationary distribution is an $n$-dimension vector $\pi=(\pi_1, \pi_2, \ldots, \pi_n)^\top=(d_1/2m, d_2/2m, \ldots, d_n/2m)^\top$.

In general, the transition  probability matrix $T$ of graph $G$ is asymmetric. However, $T$ is similar to a symmetric matrix   $P$ defined as
\begin{equation*}
	P=D^{-\frac{1}{2}} A D^{-\frac{1}{2}}=D ^{\frac{1}{2}}T D^{-\frac{1}{2}}\,,
\end{equation*}
which is often called  the normalized adjacency matrix  of $G$. By definition,   the $(i,j)$th entry of $P$ is $P(i,j)=\frac{A(i,j)}{\sqrt{d_id_j}}$. Then, it is obvious that $P(i,j)=P(j,i)$. Let $I$ be the identity matrix of approximate dimensions. Then,  $I-P$ is the normalized Laplacian matrix~\cite{Ch97} of graph $G$.

%\begin{lemma}\label{MinEig}\emph{\cite{Ch97}}
%	Let $G$ be a simple connected graph with $n$ nodes, then the eigenvalues of its normalized adjacency matrix $P$ can be listed as $1=\lambda_1>\lambda_2\geq\ldots\geq\lambda_n\geq-1$.
%\end{lemma}

Let  $\lambda_1$, $\lambda_2$, $\ldots$, $\lambda_n$ be the $n$ eigenvalues of matrix  $P$. Then, these $n$ eigenvalues can be listed in decreasing order as $1=\lambda_1>\lambda_2\geq \ldots \geq\lambda_n\geq-1$, with $\lambda_n=-1$ if and only if $G$ is a bipartite graph. Let $v_1$, $v_2$, $\ldots$, $v_n$ be the orthonormal eigenvectors corresponding to  the $n$ eigenvalues $\lambda_1$, $\lambda_2$,$\ldots$, $\lambda_n$, where $v_i=(v_{i1},v_{i2},\ldots,v_{in})^\top$, $i=1,2,\ldots, n$. Then,
\begin{equation}\label{eig=1}
	v_1=\Big(\sqrt{d_1/2m},\sqrt{d_2/2m},...,\sqrt{d_n/2m}\Big)^\top
\end{equation}
and
\begin{equation}
	\sum_{k=1}^n v_{ik}v_{jk}=\sum_{k=1}^n v_{ki}v_{kj}=\left\{
	\begin{array}{ll}
		1, & \hbox{if $i=j$,} \\
		0, & \hbox{otherwise.}
	\end{array}
	\right.
\end{equation}

A key quantity associated with random walks is  hitting time. The hitting time $T_{ij}$ from one node $i$ to another node $j$ is defined as the expected number of jumps needed for a walker starting from node $i$ to reach node $j$ for the first time. The hitting time  $T_{ij}$  is encoded in the eigenvalues and eigenvectors of the normalized adjacency (or Laplacian) matrix $P$ for graph $G$.
\begin{lemma}\label{HitTime}\emph{\cite{Lo93}}
	For random walks on a simple connected graph $G$ with $n$ nodes and $m$ edges,  the  hitting time $T_{ij}$ from one node $i$ to another node $j$ can be expressed in terms of the  eigenvalues and their orthonormal eigenvectors for the normalized adjacency matrix $P$ as
	\begin{equation*}
		T_{ij}=2m\sum_{k=2}^{n} \frac{1}{1-\lambda_k}
		\left(\frac{v_{kj}^2}{d_j}-\frac{v_{ki}v_{kj}}{\sqrt{d_i d_j}}\right).  %, \forall i,j \in V(G)
	\end{equation*}
\end{lemma}
The hitting time is relevant in various scenarios~\cite{Re01}. For example, it has been used to design clustering algorithm~\cite{ChLiTa08,Ab18},  to measure the transmission costs in wireless networks~\cite{LiZh13IEEE, ElMaPr06}, as well as  to evaluate the centrality of nodes in complex networks~\cite{WhSm03,ZhXuZh20}.

For a graph $G$,  $T_{ij}$  is usually not equal to  $T_{ji}$.  However, the  commute time between a pair of nodes can make up for this shortcoming.  For two nodes $i$ and $j$, their commute time   $C_{ij}$ is  defined as  the sum of $T_{ij}$ and $T_{ji}$, namely, $C_{ij}=T_{ij}+T_{ji}$. Thus, the relation $C_{ij}=C_{ji}$ always holds for any pair of nodes nodes $i$ and $j$.
\begin{lemma}\emph{\cite{ChRaRuSmTi97}} \label{Foster}
	Let $G$ be a simple connected graph with $n$ nodes and $m$ edges. Then the sum of commute times $C_{ij}$ between all the $m$ pairs of adjacent nodes in  $G$ is equivalent to $2m(n-1)$, i.e.
	\begin{equation*}
		\sum_{(i,j)\in E}C_{ij}=2m(n-1).
	\end{equation*}
\end{lemma}
The symmetry of commute time makes it  have many applications in different areas, such as link prediction~\cite{FoPiReSa07} and graph embedding~\cite{CuWaPeZh19}.
In addition to commute time, many other interesting quantities of graph $G$ can be defined or derived from hitting times.  For example, the mean hitting time $\bar{H}(G)$ of a graph $G$ with $n$ nodes is the
average of hitting times over all $n(n-1)$ node pairs:
\begin{equation}\label{meanht}
	\bar{H}(G)=\frac{1}{n(n-1)}\sum_{i=1}^{n}\sum_{j=1}^{n}T_{ij}.
\end{equation}
The quantity of mean hitting time has  been utilized as  an indicator of mean cost of search in networks~\cite{GuDiVeCaAr02,FeQuYi14} and  global utility of social recommender networks~\cite{WoLiChMiLiCh16}.

Another quantity defined according to hitting times is the  Kemeny's constant. For a graph $G$, its Kemeny's constant $K(G)$ is defined as  the expected number of steps required for a walker starting from a node $i$ to a destination node  chosen randomly according to a stationary distribution of  random walks on $G$~\cite{Hu14}, that is $K(G)=\sum_{j=1}^{n}\pi_j T_{ij}$.   The Kemeny constant $K(G)$  is independent of the selection of starting node $i$~\cite{LeLo02}, which means $\sum_{j=1}^{n}\pi_j T_{ij}=\sum_{j=1}^{n}\pi_j T_{kj}$ holds for an arbitrary pair of
node $i$ and $k$. Interesting, the Kemeny's constant of graph $G$ is only dependent on the eigenvalues of matrix $P$.
\begin{lemma}\emph{\cite{butler2016algebraic}}\label{lemmaKem}
	Let $G$ be a simple connected graph with $n$ nodes. Then, the Kemeny's constant $K(G)$ of  $G$  can be represented as
	\begin{equation}
		K(G)=\sum_{j=1}^{n}\pi_j T_{ij}= \sum_{k=2}^n\frac{1}{1-\lambda_k}..
	\end{equation}
	%where $1=\lambda_1>\lambda_2\geq\ldots\geq\lambda_n\geq-1$ are eigenvalues of matrix $P$.
\end{lemma}

The Kemeny constant has also found applications in diverse areas~\cite{Hu14}. It has been widely used to characterize the criticality~\cite{DeMeRoSaVa18,LeSa18} or connectivity~\cite{BeHe19} for a graph. Moreover, it can be applied to measure the efficiency of user navigation through the World Wide Web~\cite{LeLo02}.  Finally, it was also exploited to quantify the performance of a class of noisy formation control protocols~\cite{JaOl19}, and to gauge the efficiency of robotic surveillance in network environments~\cite{PaAgBu15}. Very recently, some properties and nearly linear time algorithms for computing the Kemeny's constant have been studied or developed~\cite{ZhXuZh20,XuShZhKaZh20}.

\section{Network Construction, Properties, and Important Matrices}

In this section, we  introduce the construction and properties for the studied networks, and  provide some relations among matrices related to the networks, which are very useful for deriving the properties of eigenvalues and eigenvectors of the normalized adjacency matrix, as well as the  hitting times.

%%%%%%%%%%%%%%%%%%%%%%%%%%%%%%%%%%%%%%%%%%%%%%%%%%%%%
%Figure 1
%%%%%%%%%%%%%%%%%%%%%%%%%%%%%%%%%%%%%%%%%%%%%%%%%%%%%
\begin{figure}
	\centering
	\includegraphics[width=1.0\linewidth]{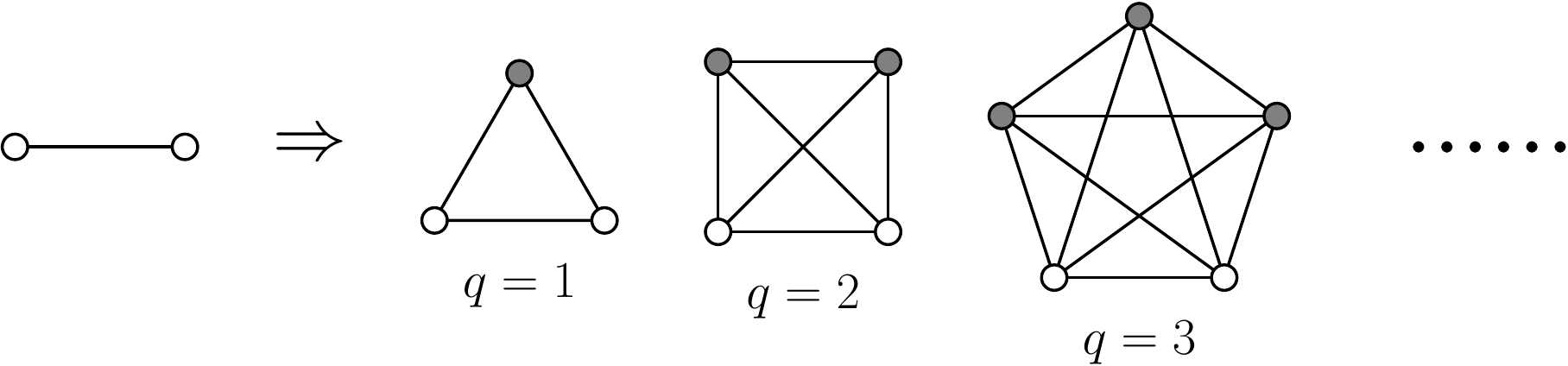}
	\caption{Network construction method. The next-iteration network is obtained from the current network by performing the operation on the right-hand side of the arrow  for each existing edge.}
	\label{build}
\end{figure}
%%%%%%%%%%%%%%%%%%%%%%%%%%%%%%%%%%%%%%%%%%%%%%%%%%%%%

\subsection{Network Construction and Properties}

The network family studied here is proposed in~\cite{WaYiXuZh19} and constructed in an iterative way. It is controlled by two parameters $q$ and $g$ with $q \geq 1$ and $g\geq 0$. Let $\mathcal{K}_q$ $(q \geq1)$ denote the complete graph with $q$ nodes. For $q = 1$,  suppose that $\mathcal{K}_1$ is a graph with an isolate node. Let $\mathcal{G}_q(g)$ be the network after $g$ iterations.  Then, $\mathcal{G}_q(g)$ is constructed as follows. For
$g = 0$, $\mathcal{G}_q(0)$ is the complete graph $\mathcal{K}_{q+2}$. For $g > 0$, $\mathcal{G}_q(g+1)$
is obtained from $\mathcal{G}_q(g)$ by performing the  operation
shown in Fig.~\ref{build}: for every existing edge of $\mathcal{G}_q(g)$, a
complete graph $\mathcal{K}_q$ is introduced, every node of which is  connected  to both
end nodes of the edge. Figures~\ref{netA} and~\ref{netB} illustrate the networks
corresponding to two particular cases of $q = 1$ and $q = 2$.

%%%%%%%%%%%%%%%%%%%%%%%%%%%%%%%%%%%%%%%%%%%%%%%%%%%%%
\begin{figure}
	\centering
	\includegraphics[width=0.45\textwidth]{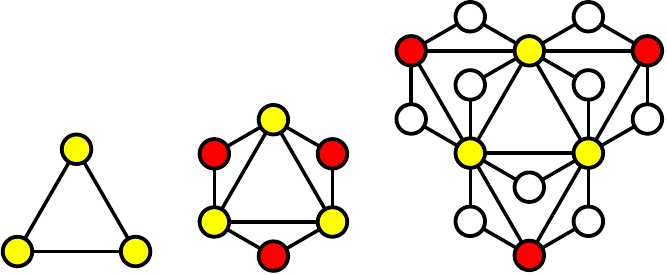}
	\caption{The networks of the first three iterations for $q =1$.} %(a) and (b) $q =2$.
	\label{netA}
\end{figure}
%%%%%%%%%%%%%%%%%%%%%%%%%%%%%%%%%%%%%%%%%%%%%%%%%%%%%

%%%%%%%%%%%%%%%%%%%%%%%%%%%%%%%%%%%%%%%%%%%%%%%%%%%%%
\begin{figure}
	\centering
	\includegraphics[width=0.45\textwidth]{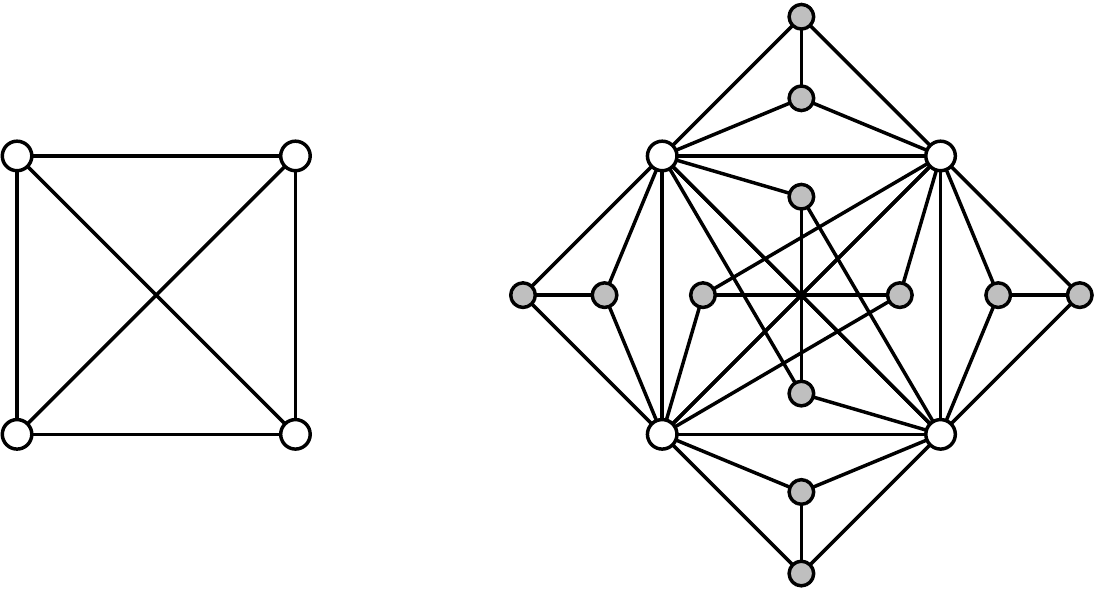}
	\caption{The networks of the first two iterations for $q =2$.} %(a) and (b) $q =2$.
	\label{netB}
\end{figure}
%%%%%%%%%%%%%%%%%%%%%%%%%%%%%%%%%%%%%%%%%%%%%%%%%%%%%

For network $\mathcal{G}_q(g)$, let $\mathcal{V}_g$ and $\mathcal{E}_g$ denote its node set and
edge set, respectively. And let $N_g=|\mathcal{V}_g|$ and $M_g=|\mathcal{E}_g|$ denote,
respectively, the number of nodes and the number of edges in graph $\mathcal{G}_q(g)$. Then, for all $g\geqslant0$,
\begin{align}
	&M_g=\left(\frac{(q+1)(q+2)}{2}\right)^{g+1},\\
	&N_g=\frac{2}{q+3}\left(\frac{(q+1)(q+2)}{2}\right)^{g+1}+\frac{2(q+2)}{q+3}.
\end{align}
The  node set $\mathcal{V}_{g+1}$ of $\mathcal{G}_q(g+1)$ can be classified into two disjoint parts $\mathcal{V}_g$ and $\mathcal{W}_{g+1}$,  where $\mathcal{V}_g$ is the set of old nodes belonging to $\mathcal{G}_q(g)$, while $\mathcal{W}_{g+1}$ is the set of new nodes generated in the process of performing aforementioned operation on $\mathcal{G}_q(g)$. Moreover, $\mathcal{W}_{g+1}$ can be further divided into $q$ disjoint subsets $\mathcal{V}^{(1)}_{g+1}$, $\mathcal{V}^{(2)}_{g+1}$, $\ldots$,  $\mathcal{V}^{(i)}_{g+1}$  satisfying $\mathcal{W}_{g+1}=\mathcal{V}^{(1)}_{g+1} \cup  \mathcal{V}^{(2)}_{g+1}\cup \cdots  \cup \mathcal{V}^{(q)}_{g+1}$,  with each $\mathcal{V}^{(i)}_{g+1}$ ($i=1$, $2$,$\ldots$, $q$) including $M_g$ new nodes produced by $M_g$ different edges in $\mathcal{G}_q(g)$. Hence, one has
\begin{equation} \label{div}
	\mathcal{V}_{g+1}  =   \mathcal{V}_{g}{\cup{\mathcal{V}^{(1)}_{g+1}}} {\cup{\mathcal{V}^{(2)}_{g+1}}} {\cup{...}} {\cup{\mathcal{V}^{(q)}_{g+1}}}.
\end{equation}
For any new node $x \in \mathcal{W}_{g+1}$, there are two old neighboring nodes in $\mathcal{V}_{g}$,  the set of which is denoted by $\Gamma(x)$. By construction, for each old edge $uv \in \mathcal{E}_{g}$, there exists one and only one node $x$ in each $\mathcal{V}^{(i)}_{g+1}$ ($i=1$, $2$, $\ldots$, $q$), satisfying  $\Gamma(x)=\{u,v\}$. Therefore, for two different sets $\mathcal{V}^{(i)}_{g+1}$ and $\mathcal{V}^{(j)}_{g+1}$, their nodes have equivalent  structural and dynamical properties.

Let $W_{g+1}=|\mathcal{W}_{g+1}|$ represent the number of those newly nodes generated at iteration $g+1$. Then,
\begin{equation}\label{W}
	W_{g+1}=q\left(\frac{(q+1)(q+2)}{2}\right)^{g+1}.
\end{equation}
Let $d^{(g)}_v$ denote the degree of a node $v$ in graph $\mathcal{G}_q(g)$,  which was generated at iteration $g_v$. Then,
\begin{gather}\label{d_v}
	d^{(g)}_v=(q+1)^{g-g_v+1}.
\end{gather}
In graph $\mathcal{G}_q(g)$, all simultaneously emerging nodes has the same degree. Thus, the number of nodes with  degree $(q+1)^{g-g_v+1}$ is equal to $q+2$ and $q\left(\frac{(q+1)(q+2)}{2}\right)^{g_v}$ for $g_v=0$ and $g_v>0$, respectively.

The resulting family of networks is consist of cliques $\mathcal{K}_{q+2}$ or smaller cliques, and are thus called simplicial networks, characterized by a parameter $q$.  These networks display some remarkable properties that are observed in most real networks~\cite{Ne03}. They are scale-free, since their node degrees obey a power-law distribution $P(d)\sim d^{-\gamma_q}$ with $\gamma_q=2 +\frac{\ln (q+2)}{\ln (q+1)}-\frac{\ln 2}{\ln (q+1)}$~\cite{WaYiXuZh19}. They are  small-world with their diameters increasing logarithmically with the number of nodes   and their mean clustering coefficients converging to a large constant $\frac{q^2+3q+3 }{q^2+3q+5}$~\cite{WaYiXuZh19}. In addition, they have a finite spectral dimension $\frac{2(\ln(q^2+3q+3)-\ln 2)}{\ln (q+1)}$.

\subsection{Relations among Various Matrices}

Let $A_g$  denote the adjacency matrix of graph $\mathcal{G}_q(g)$. The element $A_g(i,j)$ at row $i$ and column $j$ of $A_g$ is defined as follows:  $A_g(i,j)=1$ if nodes $i$ and $j$ are directly connected by an edge in $\mathcal{G}_q(g)$, $A_g(i,j)=0$ otherwise. Let $B_g$ denote the incidence matrix of graph $\mathcal{G}_q(g)$. The element $B_g(i,j)$ at row $i$ and column $j$ of $B_g$ is: $B_g(i,j)=1$ if node $i$ is incident with edge $e_j$ in
$\mathcal{G}_q(g)$, $B_g(i,j)=0$ otherwise. Let $D_g$ denote the diagonal degree matrix of matrix graph $\mathcal{G}_q(g)$, with the $i$th diagonal element being the degree $d_i^{(g)}$ of node $i$. And let $P_{g}=D_g^{-\frac{1}{2}}A_gD_g^{-\frac{1}{2}}$ denote the normalized adjacency matrix of graph $\mathcal{G}_q(g)$. Then
for $\mathcal{G}_q(g+1)$,  its adjacency matrix $A_{g+1}$,  diagonal degree matrix $D_{g+1}$ and normalized adjacency matrix $P_{g+1}$, can be expressed in terms of related matrices of $\mathcal{G}_q(g)$ as
\begin{equation*}
	A_{g+1}=
	\left(
	\begin{array}{ccccc}
		A_g           & B_g     &B_g    & \cdots\   & B_g \\
		B_g^\top      & O       &I  & \cdots\   & I  \\
		B_g^\top      & I       &O  & \cdots\   & I  \\
		\vdots\     & \vdots\  &\vdots  & \ddots\   & \vdots\ \\
		B_g^\top    &I  & I         &\cdots\    & O  \\
	\end{array}
	\right),
\end{equation*}
\begin{equation*}
	D_{g+1}={\rm diag} \{(q+1)D_g,\underbrace{(q+1)I,...,(q+1)I}_q\},
\end{equation*}
and
\begin{small}
	\begin{align}\label{MatrixP}
		&P_{g+1}   =   D_{g+1}^{-\frac{1}{2}} A_{g+1} D_{g+1}^{-\frac{1}{2}}\\
		=&\frac{1}{q+1}
		\left(
		\begin{array}{ccccc}
			P_g    & D_g^{-\frac{1}{2}}B_g &D_g^{-\frac{1}{2}}B_g & \cdots\   & D_g^{-\frac{1}{2}}B_g \\
			B_g^\top D_g^{-\frac{1}{2}} & O   & I              & \cdots\   & I \\
			B_g^\top D_g^{-\frac{1}{2}} & I   & O              & \cdots\   & I \\
			\vdots\                 & \vdots\    &\vdots       & \ddots\   & I \\
			B_g^\top D_g^{-\frac{1}{2}} & I  &I                 & \cdots\   & O \\
		\end{array}
		\right)\notag\,.
	\end{align}
\end{small}
%where $I$ is the $M_g \times M_g $ identity matrix.

\section{Eigenvalues and Eigenvectors of Normalized Adjacency Matrix}

In this section, we study the eigenvalues and eigenvectors of normalized adjacency matrix $P_{g+1}$  for graph $\mathcal{G}_q(g+1)$,  expressing both eigenvalues and eigenvectors for $P_{g+1}$  in terms of those associated with graph $\mathcal{G}_q(g)$. Then we  use these results to obtain the Kemeny's constant of graph $\mathcal{G}_q(g)$.

For the purpose of analyzing the eigenvalues and eigenvectors of matrix $P_{g+1}$  precisely, we first study the orthonormal basis $\mathcal{Y}_g$ of the kernel space of the following matrix
\begin{equation*}
	C_g:={\underbrace{\left(
			\begin{array}{cccc}
				B_g & B_g & \cdots\ & B_g  \\
			\end{array}
			\right)}_q}.
\end{equation*}
Since $\mathcal{G}_q(g)$ is non-bipartite, by Lemma~\ref{bRank} one has ${\rm rank}(B_g)=N_g$ Thus, $\dim({\rm Ker}(B_g))=M_g-N_g$, ${\rm rank}(C_g)=N_g$, and $\dim({\rm Ker}(C_g))=qM_g-N_g$. Then, $\mathcal{Y}_g$ can be classified into two non-overlapping parts $\mathcal{Y}_g^{(1)}$ and $\mathcal{Y}_g^{(2)}$ obeying $\mathcal{Y}_g=\mathcal{Y}_g^{(1)}\cup\mathcal{Y}_g^{(2)}$, where $\mathcal{Y}_g^{(1)}$ has $M_g-N_g$  vectors, while $\mathcal{Y}_g^{(2)}$ has $(q-1)M_g$ vectors. Moreover, as will shown below, $\mathcal{Y}_g^{(1)}$ and $\mathcal{Y}_g^{2}$ can be  constructed, respectively, by using the orthonormal basis vectors of the kernel space of matrix $B_g$ and the column vectors of the $M_g\times M_g$ identity matrix $I$.

Let $\mathcal{X}_g=\{X_1(g),X_2(g),\ldots,X_{M_g-N_g}(g)\}$ denote the orthonormal basis of the kernel space of matrix $B_g$, and let $Z_i(g)$ denote the $i$th column vector of the $M_g\times M_g$ identity matrix $I$. Then the vectors in $\mathcal{Y}_g^{(1)}$ are
\begin{small}
	\begin{equation*}
		\begin{aligned}
			\frac{1}{\sqrt{q}}\left(
			\begin{array}{c}
				\!X_1^{\top}(g)\!\\
				\!X_1^{\top}(g)\!\\
				\!\vdots\!\\
				\!X_1^{\top}(g)\!\\
			\end{array}
			\right),
			\frac{1}{\sqrt{q}}\left(
			\begin{array}{c}
				X_2^{\top}(g)\\
				X_2^{\top}(g)\\
				\vdots\\
				X_2^{\top}(g)\\
			\end{array}
			\right),
			\cdots
			\frac{1}{\sqrt{q}}\left(
			\begin{array}{c}
				\!X_{M_g-N_g}^{\top}(g)\!\\
				\!X_{M_g-N_g}^{\top}(g)\!\\
				\!\vdots,\!\\
				\!X_{M_g-N_g}^{\top}(g)\!\\
			\end{array}
			\right),
		\end{aligned}
	\end{equation*}
\end{small}
and the vectors in $\mathcal{Y}_g^{(2)}$ are
\begin{small}
	\begin{equation*}
		\begin{aligned}
			\left(
			\begin{array}{c}
				\frac{1}{\sqrt{2}}Z_i(g)\\
				-\frac{1}{\sqrt{2}}Z_i(g)\\
				0\\
				0\\
				\vdots\\
				0
			\end{array}
			\right),
			\left(
			\begin{array}{c}
				\frac{1}{\sqrt{6}}Z_i(g)\\
				\frac{1}{\sqrt{6}}Z_i(g)\\
				-\frac{1}{\sqrt{3}}Z_i(g)\\
				0\\
				\vdots\\
				0
			\end{array}
			\right),\cdots,
			\left(
			\begin{array}{c}
				\frac{1}{\sqrt{q(q-1)}}Z_i(g)\\
				\frac{1}{\sqrt{q(q-1)}}Z_i(g)\\
				\frac{1}{\sqrt{q(q-1)}}Z_i(g)\\
				\vdots\\
				\frac{1}{\sqrt{q(q-1)}}Z_i(g)\\
				-\sqrt{\frac{q-1}{q}}Z_i(g)
			\end{array}
			\right),
		\end{aligned}
	\end{equation*}
\end{small}
where $i=1,2,\ldots,M_g$.

Considering the process of  network construction, we have the following lemmas.
\begin{lemma}\label{y2property}
	For any vector
	\begin{equation*}
		Y_i^{(2)}(g)=\left(
		\begin{array}{c}
			Y_{i1}^{(2)}(g)\\
			Y_{i2}^{(2)}(g)\\
			\vdots\\
			Y_{iq}^{(2)}(g)
		\end{array}
		\right), i=1,2,\ldots,(q-1)M_g,
	\end{equation*}
	in $\mathcal{Y}_g^{(2)}$, its components obey the following relation
	\begin{equation}
		Y_{i1}^{(2)}(g)+Y_{i2}^{(2)}(g)+\cdots+Y_{iq}^{(2)}(g)=0.
	\end{equation}
\end{lemma}
\begin{lemma}\label{y2property2}
	For any integer $j\in\{1,2,\ldots,qM_g\}$ and $\mathcal{Y}_g^{(2)}=\{Y_{1}^{(2)}(g),Y_{2}^{(2)}(g),\ldots,$
	$Y_{(q-1)M_g}^{(2)}(g)\}$, we have
	\begin{equation}
		\sum_{i=1}^{(q-1)M_g}\left(Y_{ij}^{(2)}(g)\right)^2=1-\frac{1}{q}.
	\end{equation}
\end{lemma}
\begin{lemma}\label{Spectra}
	Let $1=\lambda_1(g)>\lambda_2(g)\geq...\geq\lambda_{N_g}(g)>-1$ be the eigenvalues of matrix $P_g$, and let $v_1(g),v_2(g),...,v_{N_g}(g)$ be their corresponding orthonormal eigenvectors. Then
	$\frac{\lambda_i(g)+q}{q+1}$, $i=1$, $2$,$\ldots$, $N_g$, are eigenvalues of matrix $P_{g+1}$, and their corresponding orthonormal eigenvectors are
	\begin{equation}\label{spec1}
		\sqrt{\frac{\lambda_i(g)+1}{q+\lambda_i(g)+1}}
		\left(
		\begin{array}{c}
			v_i(g) \\
			\frac{1}{\lambda_i(g)+1} B_g^\top D_g^{-\frac{1}{2}}v_i(g) \\
			\vdots \\
			\frac{1}{\lambda_i(g)+1} B_g^\top D_g^{-\frac{1}{2}}v_i(g) \\
		\end{array}
		\right);
	\end{equation}
	$-\frac{1}{q+1}$'s are eigenvalues of matrix $P_{g+1}$ with multiplicity $(q-1)M_g+N_g$, and the corresponding orthonormal eigenvectors are
	\begin{equation}\label{spec2}
		\sqrt{\frac{q}{q+\lambda_i(g)+1}}
		\left(
		\begin{array}{c}
			v_i(g) \\
			-\frac{1}{q} B_g^\top D_g^{-\frac{1}{2}}v_i(g) \\
			\vdots \\
			-\frac{1}{q} B_g^\top D_g^{-\frac{1}{2}}v_i(g) \\
		\end{array}
		\right),
	\end{equation}
	$i=1$, $2$,$\ldots$, $N_g$, and
	\begin{equation}\label{spec3}
		\\ \left(
		\begin{array}{c}
			0 \\
			Y_z^{(2)}(g) \\
		\end{array}
		\right) ,~z=1,~2,\ldots,~(q-1)M_g,
	\end{equation}
	where $Y_z^{(2)}(g)\in\mathcal{Y}_g^{(2)}$;
	and $\frac{q-1}{q+1}$'s are eigenvalues of $P_{g+1}$ having multiplicity $M_g-N_g$,  with the corresponding orthonormal eigenvectors being
	\begin{equation}\label{spec4}
		\\ \left(
		\begin{array}{c}
			0 \\
			Y_z^{(1)}(g) \\
		\end{array}
		\right) ,~z=1,~2,\ldots,~M_g-N_g,
	\end{equation}
	where $Y_z^{(1)}(g)\in\mathcal{Y}_g^{(1)}$.
\end{lemma}
\begin{proof}
	%Since $\mathcal{G}_q(g)$ is non-bipartite, by Lemma~\ref{MinEig}, every  eigenvalue $\lambda_i(g)$ of $P_g$ is not equal to $-1$.
	For any eigenpair $\lambda_i(g)$ and $v_i(g)$,  $P_gv_i(g)=\lambda_i(g) v_i(g)$ holds. Then, by Lemma~\ref{BBAD} and Eq.~\eqref{MatrixP}, one has
	\begin{small}
		\begin{equation*}
			\begin{aligned}
				&P_{g+1}\!
				\left(
				\begin{array}{c}
					\!\!v_i(g)\!\! \\
					\!\!\frac{1}{\lambda_i(g)+1} B_g^\top D_g^{-\frac{1}{2}}v_i(g)\!\! \\
					\!\!\vdots \!\!\\
					\!\!\frac{1}{\lambda_i(g)+1} B_g^\top D_g^{-\frac{1}{2}}v_i(g)\!\! \\
				\end{array}
				\right)
				\!=\!
				\frac{1}{q\!+\!1}\!\left(
				\begin{array}{c}
					\!\!\lambda_i(g)v_i(g)+qv_i(g) \!\!\\
					\!\!\frac{\lambda_i(g)+q}{\lambda_i(g)+1} B_g^\top D_g^{-\frac{1}{2}}v_i(g) \!\!\\
					\!\!\vdots \!\!\\
					\frac{\lambda_i(g)+q}{\lambda_i(g)+1} B_g^\top \!\!D_g^{-\frac{1}{2}}v_i(g) \!\!\\
				\end{array}
				\right)\\
				&=\frac{\lambda_i(g)+q}{q+1}
				\left(
				\begin{array}{c}
					v_i(g) \\
					\frac{1}{\lambda_i(g)+1} B_g^\top D_g^{-\frac{1}{2}}v_i(g) \\
					\vdots \\
					\frac{1}{\lambda_i(g)+1} B_g^\top D_g^{-\frac{1}{2}}v_i(g) \\
				\end{array}
				\right)
			\end{aligned}
		\end{equation*}
	\end{small}
	and
	\begin{small}
		\begin{equation*}
			\begin{aligned}
				&P_{g+1}
				\left(
				\begin{array}{c}
					v_i(g) \\
					-\frac{1}{q} B_g^\top D_g^{-\frac{1}{2}}v_i(g) \\
					\vdots \\
					-\frac{1}{q} B_g^\top D_g^{-\frac{1}{2}}v_i(g) \\
				\end{array}
				\right)
				=
				\left(
				\begin{array}{c}
					\frac{\lambda_i(g)v_i(g)-(\lambda_i(g)+1)v_i(g)}{q+1} \\
					\frac{1}{q(q+1)} B_g^\top D_g^{-\frac{1}{2}}v_i(g) \\
					\vdots \\
					\frac{1}{q(q+1)} B_g^\top D_g^{-\frac{1}{2}}v_i(g) \\
				\end{array}
				\right)\\
				=&
				-\frac{1}{q+1}
				\left(
				\begin{array}{c}
					v_i(g) \\
					-\frac{1}{q} B_g^\top D_g^{-\frac{1}{2}}v_i(g) \\
					\vdots \\
					-\frac{1}{q} B_g^\top D_g^{-\frac{1}{2}}v_i(g) \\
				\end{array}
				\right),
			\end{aligned}
		\end{equation*}
	\end{small}
	both of which  lead to~\eqref{spec1} and~\eqref{spec2} through normalization.
	
	In addition, according to Lemma~\ref{y2property}, one has
	\begin{small}
		\begin{align*}
			P_{g+1}\left(
			\begin{array}{c}
				0 \\
				Y_z^{(2)}(g) \\
			\end{array}
			\right)&=\frac{1}{q+1}\left(
			\begin{array}{c}
				0 \\
				-Y_{i1}^{(2)}(g)+\sum_{k=1}^qY_{ik}^{(2)}(g) \\
				-Y_{i2}^{(2)}(g)+\sum_{k=1}^qY_{ik}^{(2)}(g) \\
				\vdots\\
				-Y_{iq}^{(2)}(g)+\sum_{k=1}^qY_{ik}^{(2)}(g) \\
			\end{array}
			\right)\\
			&=\frac{1}{q+1}\left(
			\begin{array}{c}
				0 \\
				-Y_{i1}^{(2)}(g) \\
				-Y_{i2}^{(2)}(g)\\
				\vdots\\
				-Y_{iq}^{(2)}(g)\\
			\end{array}
			\right)\\
			&=-\frac{1}{q+1}\left(
			\begin{array}{c}
				0 \\
				Y_z^{(2)}(g) \\
			\end{array}
			\right),
		\end{align*}
	\end{small}	
	as claimed by~\eqref{spec3}.
	
	Finally, for each $z=1,~2,\ldots,~(q-1)M_g$,
	\begin{small}
		\begin{align*}
			P_{g+1}\left(
			\begin{array}{c}
				0 \\
				Y_z^{(1)}(g) \\
			\end{array}
			\right)
			&=P_{g+1}\left(
			\begin{array}{c}
				X_i(g) \\
				X_i(g) \\
				\vdots\\
				X_i(g)\\
			\end{array}
			\right)\\
			&=\frac{1}{q+1}\left(
			\begin{array}{c}
				0 \\
				(q-1)X_{i}(g) \\
				(q-1)X_{i}(g)\\
				\vdots\\
				(q-1)X_{i}(g)\\
			\end{array}
			\right)\\
			&=\frac{q-1}{q+1}\left(
			\begin{array}{c}
				0 \\
				Y_z^{(2)}(g) \\
			\end{array}
			\right).
		\end{align*}
	\end{small}	
	Thus, we complete the proof.
\end{proof}

In fact, the orthonormal eigenvectors of $\mathcal{G}_q(g+1)$  can be expressed in more explicit forms. By Eq.~\eqref{eig=1} and Lemma~\ref{Spectra}, one can easily derive the following results.
\begin{corollary}\label{corospec}
	Let $1=\lambda_1(g)>\lambda_2(g)\geq...\geq\lambda_{N_g}(g)> -1$ be the eigenvalues of matrix $P_g$, and let $v_1(g),v_2(g),...,v_{N_g}(g)$ be their corresponding orthonormal eigenvectors. Then,
	\begin{enumerate}
		\item The eigenvectors corresponding to eigenvalues $\frac{\lambda_1(g)+q}{q+1}=1$, the first $-\frac{1}{q+1}$ for matrix $P_{g+1}$ are
		\begin{equation}\label{re1A}
			\begin{aligned}
				&\left(\sqrt{\frac{d_1(g)}{M_g(q+2)}}, \cdots, \sqrt{\frac{d_{N_g}(g)}{M_g(q+2)}},\right.\\
				&\left.\frac{1}{\sqrt{M_g(q+2)}},  \cdots, \frac{1}{\sqrt{M_g(q+2)}} \right)^\top
			\end{aligned}
		\end{equation}
		and
		\begin{equation}\label{re1B}
			\begin{aligned}
				&\left(\sqrt{\frac{qd_1(g)}{2M_g(q+2)}}, \cdots, \sqrt{\frac{qd_{N_g}(g)}{2M_g(q+2)}},\right.\\
				&\left.-\sqrt{\frac{2}{qM_g(q+2)}}, \cdots,-\sqrt{\frac{2}{qM_g(q+2)}}\right)^\top,
			\end{aligned}
		\end{equation}
		respectively.
		\item
		The element of orthonormal eigenvectors for eigenvalues $\frac{\lambda_i(g)+q}{q+1}$, $i=2,3,\ldots,N_g$, corresponding to node $j$ is
		\begin{equation*}
			\left\{
			\begin{array}{ll}
				\!\!\!\sqrt{\frac{\lambda_i(g)+1}{\lambda_i(g)+q+1}}v_{ij}(g)&j\in \mathcal{V}_g\\
				\!\!\!\sqrt{\frac{1}{(\lambda_i(g)+1)(\lambda_i(g)+q+1)}}\left(\frac{v_{is}(g)}{\sqrt{d_s(g)}}\!+\!\frac{v_{it}(g)}{\sqrt{d_t(g)}}\right)&\!\!\!\!\!j\in \mathcal{W}_{g+1};
			\end{array}\right.
		\end{equation*}
		and the element of orthonormal eigenvectors for eigenvalues $-\frac{1}{q+1}$, $i=2,3,\ldots,N_g$, corresponding to node $j$ is
		\begin{equation*}
			\left\{\begin{array}{ll}
				\sqrt{\frac{q}{\lambda_i(g)+q+1}}v_{ij}(g),&j\in \mathcal{V}_g,\\
				\sqrt{\frac{1}{q(\lambda_i(g)+q+1)}}\Big(\frac{v_{is}(g)}{\sqrt{d_s(g)}}+\frac{v_{it}(g)}{\sqrt{d_t(g)}}\Big),&j\in \mathcal{W}_{g+1}
			\end{array}\right.
		\end{equation*}
		where $\Gamma(j)=\{s,t\}$;
		\item For orthonormal eigenvectors $
		\left(
		\begin{array}{c}
			0 \\
			Y_z(g) \\
		\end{array}
		\right)
		$, $z=1,2,\ldots,qM_g-N_g$, of eigenvalues $0's$ of matrix $P_{g+1}$, we have
		\begin{small}
			\begin{align}\label{sum1}
				\sum_{z=1}^{qM_g-N_g}Y_{zj}^{2}(g)=&1-\frac{1}{qM_g}-\sum_{k=2}^{N_g} \frac{1}{(1+\lambda_k(g))q}\\
				&\quad\quad\quad\quad\quad\quad~\left(\frac{v_{ks}(g)}{\sqrt{d_s(g)}}+\frac{v_{kt}(g)}{\sqrt{d_t(g)}}\right)^2\notag.
			\end{align}
		\end{small}
		for each $j\in\mathcal{W}_{g+1}$ with $\Gamma(j)=\{s,t\}$.
	\end{enumerate}
\end{corollary}

\section{Two-Node Hitting Time and  Kemeny's Constant}

Lemma~\ref{Spectra} and Corollary~\ref{corospec} provide  complete information about the eigenvalues and eigenvectors of matrix $P_{g+1}$ in terms of those of matrix $P_g$ of the previous iteration. In this section, we use this information to determine two-node hitting time and the Kemeny's constant for unbiased random walks on graph $\mathcal{G}_q(g+1)$.

\subsection{Two-Node Hitting Time}

We first present our results about  hitting times for random walks on graph $\mathcal{G}_q(g)$. Let $T_{ij}(g)$ denote the hitting time from node $i$ to node $j$ in $\mathcal{G}_q(g)$.
\begin{theorem}\label{qTriHT}
	For networks $\mathcal{G}_q(g)$ and $\mathcal{G}_q(g+1)$,
	\begin{enumerate}
		\item if $i$, $j\in \mathcal{V}_g$, then $T_{ij}(g+1)=(q+1)T_{ij}(g)$;
		\item if $i\in \mathcal{W}_{g+1}$, $j\in \mathcal{V}_g$, $\Gamma(i)=\{s,t\}$, then
		\begin{small}
			\begin{equation*}
				\begin{aligned}
					T_{ij}(g+1)=& \frac{q+1}{2}+\frac{q+1}{2}\left(T_{sj}(g)+T_{tj}(g)\right),\\
					T_{ji}(g+1)=&\frac{3(q+1)}{2}M_g-\frac{q+1}{2}+\frac{q+1}{4}\\
					&\cdot\left(2\left(T_{js}(g)+T_{jt}(g)\right)-\left(T_{ts}(g)+T_{st}(g)\right)\right);
				\end{aligned}
			\end{equation*}
		\end{small}
		
		\item if $i$, $j\in \mathcal{W}_{g+1}$, (a) $i$ is adjacent to $j$, then
		\begin{small}
			\begin{equation*}
				T_{ij}(g+1) = (q+1)M_g,
			\end{equation*}
		\end{small}
		(b) else if $i$ is not adjacent to $j$, $\Gamma(i)=\{s,t\}$, and $\Gamma(j)=\{u,v\}$, then
		\begin{small}
			\begin{equation*}
				\begin{aligned} T_{ji}(g+1)=&\frac{3(q+1)}{2}M_g+\frac{q+1}{4}\left(T_{su}(g)+T_{tu}(g)\right.\\&\left.+T_{sv}(g)+T_{tv}(g)-(T_{uv}(g)+T_{vu}(g))\right).
				\end{aligned}
			\end{equation*}
		\end{small}
	\end{enumerate}
	%where $\tilde{T}_{ij}$ denotes hitting times on graph $\mathcal{G}_q(g)$, while $T_{ij}$ denotes $G$'s.
\end{theorem}
\begin{proof}
	Note that $M_{g+1}=\frac{(q+1)(q+2)}{2}M_g$, $d_i(g+1)=(q+1)d_i(g)$ for $i\in \mathcal{V}_g$, and $d_i(g+1)=2$ for $i\in \mathcal{W}_{g+1}$.
	
	We first prove  1). By Lemmas~\ref{HitTime} and~\ref{Spectra}, one has
	\begin{small}
		\begin{equation*}
			\begin{aligned}
				&T_{ij}(g+1)\\
				=&2M_q(g+1)\sum_{k=2}^{N_g}
				\Bigg(\frac{1}{1-\frac{\lambda_k(g)+q}{q+1}}
				\frac{\lambda_k(g)+1}{\lambda_k(g)+q+1}
				\\
				&\quad\quad\quad\quad\quad\quad\quad~+\frac{1}{1+\frac{1}{q+1}}
				\frac{q}{\lambda_k(g)+q+1}\Bigg)\\&
				\quad\quad\quad\quad\quad\quad~~~\left(\frac{v_{kj}(g)^2}{(q+1)d_j(g)}-\frac{v_{ki}(g)v_{kj}(g)}{(q+1)\sqrt{d_i(g)d_j(g)}}\right)\\
				=&  2M_g\frac{(q+1)(q+2)}{2}\sum_{k=2}^{N_g}\frac{2q+2}{q+2}\frac{1}{1-\lambda_k(g)}\\
				&\quad\quad\quad\quad\quad\quad\quad\quad\quad~~\left(\frac{v_{kj}(g)^2}{(q+1)d_j(g)}-\frac{v_{kj}(g)v_{ki}(g)}{(q+1)\sqrt{d_i(g)d_j(g)}}\right)\\
				=&   (q+1)\cdot2M_g\sum_{k=2}^{N_g}\frac{1}{1-\lambda_k(g)}
				\left(\frac{v_{kj}(g)^2}{d_j(g)}-\frac{v_{kj}(g)v_{ki}(g)}{\sqrt{d_i(g)d_j(g)}}\right)\\
				=&    (q+1)T_{ij}(g).
			\end{aligned}
		\end{equation*}
	\end{small}
	Thus 1) is proved.
	
	%If $G$ is non-bipartite, then by Theorem \ref{HitTime} and Eqs.(\ref{re1A})(\ref{re3}), we have
	%\begin{equation*}
	%    \begin{aligned}
		%        \tilde{T}_{ij}=& \frac{2(q+1)}{q+2}\cdot2m\sum_{k=2}^{n-1}\frac{1}{1-\lambda_k(g)}
		%                            (\frac{v_{kj}(g)^2}{d_j(g)}-\frac{v_{kj}(g)v_{ki}}{\sqrt{d_id_j(g)}})
		%                        \\&+2(m+2mq)\frac{1}{1+\frac{1}{q+1}}
		%                            \Big(\frac{v_{nj}^2}{(q+1)d_j(g)}-\frac{v_{nj}v_{ni}}{(q+1)\sqrt{d_id_j(g)}}\Big)
		%                        \\=&\frac{4q+2}{q+2}\cdot2m\sum_{k=2}^{N_g}\frac{1}{1-\lambda_k(g)}
		%                            (\frac{v_{kj}(g)^2}{d_j(g)}-\frac{v_{kj}(g)v_{ki}}{\sqrt{d_id_j(g)}})
		%                            =\frac{4q+2}{q+2}T_{ij}.
		%    \end{aligned}
	%\end{equation*}
	%where the next-to-last equality is obtained by using the fact that $\lambda_n(g)=-1$. Thus (1) is proved.
	
	We continue to prove 2). Since $\Gamma(i)=\{s,t\}$,
	\begin{equation*}
		\begin{aligned}
			T_{ij}(g+1)=&\frac{1}{q+1}\left(1+T_{sj}(g+1)+1+T_{tj}(g+1)\right.\\&\left.\quad\quad~+(q-1)(1+T_{ij}(g+1))\right)\\
			=&\frac{q+1}{2}+\frac{1}{2}\left(T_{sj}(g+1)+T_{tj}(g+1)\right)\\
			=&\frac{q+1}{2}+\frac{q+1}{2}\left(T_{sj}(g)+T_{tj}(g)\right).\\
		\end{aligned}
	\end{equation*}
	While for $T_{ji}(g+1)$, by Lemmas~\ref{HitTime},~\ref{y2property2},~\ref{Spectra} and Corollary~\ref{corospec}, one obtains
	\begin{small}
		\begin{equation*}
			\begin{aligned}
				&T_{ji}(g+1)\\
				=&  2M_q(g+1)\Bigg(\frac{1}{1+\frac{1}{q+1}}\frac{1}{q(q+1)M_g}
				+\sum_{k=2}^{N_g}\frac{1}{q+1}\\
				&\left(\frac{v_{ks}(g)}{\sqrt{d_s(g)}}\!\!+\!\!\frac{v_{kt}(g)}{\sqrt{d_t(g)}}\right)^2\!\!   \Bigg(\!\!\left(\frac{1}{1\!\!-\!\!\frac{\lambda_k(g)\!+\!q}{q\!+\!1}}
				\frac{1}{(\lambda_k(g)\!\!+\!\!1)(\lambda_k(g)\!\!+\!\!q\!\!+\!\!1)}\right)\\
				&+\left(\frac{1}{1-\frac{1}{q+1}}
				\frac{1}{q(\lambda_k(g)+q+1)}\right)\Bigg)-\sum_{k=2}^{N_g}\frac{v_{kj}(g)}{\sqrt{(q+1)^2d_j(g)}}\\
				&\left(\frac{v_{ks}(g)}{\sqrt{d_s(g)}}\!\!+\!\!\frac{v_{kt}(g)}{\sqrt{d_t(g)}}\right)\!\Bigg(\!\!\Bigg(\frac{1}{1\!\!-\!\!\frac{\lambda_k(g)\!+\!q}{q\!+\!1}}
				\sqrt{\frac{1}{(\lambda_k(g)\!\!+\!\!1)(\lambda_k(g)\!\!+\!\!q\!\!+\!\!1)}}\\
				&\sqrt{\frac{\lambda_k(g)\!+\!1}{\lambda_k(g)\!+\!q\!+\!1}}\Bigg)\!-\!\left(\!\frac{1}{1\!+\!\frac{1}{q\!+\!1}}\sqrt{\frac{1}{q(\lambda_k(g)\!\!+\!q\!\!+\!\!1)}}\sqrt{\frac{q}{\lambda_k(g)\!\!+\!\!q\!\!+\!\!1}}\right)\!\!\Bigg)\\
				&+\frac{3(q+1)}{2(q+2)}-\frac{1}{2qM_g}-\sum_{k=2}^{N_g} \frac{1}{2q(1\!+\!\lambda_k(g))}\bigg(\frac{v_{ks}(g)}{\sqrt{d_s(g)}}\!+\!\frac{v_{kt}(g)}{\sqrt{d_t(g)}}\bigg)^2\\
				=&\frac{3(q+1)}{2}M_g-\frac{q+1}{2}+(q+1)M_g\sum_{k=2}^{N_g}\frac{1}{1-\lambda_k(g)}\\
				&\Bigg(\left(\frac{v_{ks}(g)^2}{d_s(g)}-\frac{v_{ks}(g)v_{kj}(g)}{\sqrt{d_s(g)d_j(g)}}\right)+\left(\frac{v_{kt}(g)^2}{d_t(g)}-\frac{v_{kt}(g)v_{kj}(g)}{\sqrt{d_t(g)d_j(g)}}\right)\\  &-\frac{1}{2}\left(\frac{v_{ks}(g)}{\sqrt{d_s(g)}}-\frac{v_{kt}(g)}{\sqrt{d_t(g)}}\right)^2\Bigg)\\
				=&\frac{3(q+1)}{2}M_g-\frac{q+1}{2}+\frac{q+1}{4}\big(2\left(T_{js}(g)+T_{jt}(g)\right)\\
				&-\left(T_{ts}(g)+T_{st}(g)\right)\big).
			\end{aligned}
		\end{equation*}
	\end{small}
	
	We finally prove 3). (a) If $i$ is adjacent to $j$, then $\Gamma(i)=\Gamma(j)=\{s,t\}$. In this case,  we obtain
	\begin{align*}
		T_{ij}(g+1)=&\frac{1}{q+1}\left(1+T_{sj}(g+1)+1+T_{tj}(g+1)\right.\\
		&\left.\quad\quad~+q-1+(q-2)T_{ij}(g+1)\right)\\
		=&\frac{q+1}{3}+\frac{1}{3}\left(T_{sj}(g+1)+T_{tj}(g+1)\right)\\
		=&\frac{q+1}{6}\left(T_{ts}(g)+T_{st}(g)-\left(T_{st}(g)+T_{ts}(g)\right)\right)\\&+(q+1)M_g\\
		=&(q+1)M_g.
	\end{align*}
	(b) If $i$ is not adjacent to $j$, considering $\Gamma(i)=\{s,t\}$, $\Gamma(j)=\{u,v\}$, we obtain
	\begin{equation*}
		\begin{aligned}
			T_{ij}(g+1)=&\frac{1}{q+1}\left(1+T_{sj}(g+1)+1+T_{tj}(g+1)\right.\\&\left.\quad\quad~+q-1+(q-1)T_{ij}(g+1)\right)\\
			=&\frac{q+1}{2}+\frac{1}{2}\left(T_{sj}(g+1)+T_{tj}(g+1)\right)\\
			=&\frac{q+1}{4}\left(T_{su}(g)+T_{tu}(g)+T_{sv}(g)+T_{tv}(g)\right.\\
			&\left.\quad\quad~-\left(T_{uv}(g)+T_{vu}(g)\right)\right)+\frac{3(q+1)}{2}M_g.
		\end{aligned}
	\end{equation*}
	This completes the proof.
\end{proof}

\subsection{Kemeny's Constant}

With Lemmas~\ref{lemmaKem} and~\ref{Spectra},  the Kemeny's constant of $\mathcal{G}_q(g)$ can be determined explicitly.
\begin{theorem}\label{Kemeny}
	Let $K_g$ be the Kemeny's constant for random walk in $\mathcal{G}_q(g)$. Then, for all $g\geq0$,
	\begin{align}\label{KemenyE}
		K_g=&\left(\frac{(q+1)^2}{q+2}-\frac{3(q+1)}{2}\right)(q+1)^g\\
		&+\frac{(q+1)(3q+7)}{2(q+3)}\left(\frac{(q+1)(q+2)}{2}\right)^g+\frac{q+1}{q+3}\notag.
	\end{align}
	When $g\to\infty$,
	\begin{equation}
		\lim_{g\to\infty}K_g=\frac{3q+7}{2(q+2)}N_g.
	\end{equation}
\end{theorem}
\begin{proof}
	Suppose that $1 = \lambda_1(g) > \lambda_2 (g)\geq\ldots \lambda_{N_g}(g) > -1$ are eigenvalues of
	the matrix $P_g$. By Lemmas~\ref{lemmaKem} and~\ref{Spectra}, we obtain
	\begin{align}\label{kr}
		K_{g+1}=&\sum_{i=2}^{N_{g+1}}\frac{1}{1-\lambda_i(g+1)}\notag\\
		=&\sum_{i=2}^{N_g}\frac{1}{1-\frac{\lambda_i(g)+q}{q+1}}+\frac{(q-1)M_g+N_g}{1+\frac{1}{q+1}}+\frac{M_g-N_g}{1-\frac{q-1}{q+1}}\notag\\
		=&(q+1)\sum_{i=2}^{N_g}\frac{1}{1-\lambda_i(g)}+\frac{3q(q+1)}{2(q+2)}M_g-\frac{q(q+1)}{2(q+2)}\notag\\
		=&(q+1)K_g+\frac{3q(q+1)}{2(q+2)}M_g-\frac{q(q+1)}{2(q+2)}
	\end{align}
	With $M_g=\left(\frac{(q+1)(q+2)}{2}\right)^{g+1}$ and the initial condition $K_0=\frac{(q+1)^2}{q+2}$, Eq.~\eqref{kr} is solved 	to obtain
	\begin{align}\label{kg}
		K_g=&\left(\frac{(q+1)^2}{q+2}-\frac{3(q+1)}{2}\right)(q+1)^g\\
		&+\frac{(q+1)(3q+7)}{2(q+3)}\left(\frac{(q+1)(q+2)}{2}\right)^g+\frac{q+1}{q+3}\notag,
	\end{align}
	which is exactly~\eqref{KemenyE}.	
	
	We continue to  express the Kemeny's constant $K_g$ in terms of
	the number of nodes $N_g$. From $N_g=\frac{2}{q+3}\left(\frac{(q+1)(q+2)}{2}\right)^{g+1}+\frac{2(q+2)}{q+3}$, we have $\left(\frac{(q+1)(q+2)}{2}\right)^g=\frac{q+3}{(q+1)(q+2)}N_g-\frac{2}{q+1}$ and $g=\ln\left(\frac{q+3}{(q+1)(q+2)}N_g-\frac{2}{q+1}\right)/\ln\left(\frac{(q+1)(q+2)}{2}\right)$. Inserting these two expressions  into Eq.~\eqref{kg}  results in
	\begin{align*}
		K_g=&\frac{q+1}{q+3}+\left(\frac{(q+1)^2}{q+2}-\frac{3(q+1)}{2}\right)\\
		&\left(\frac{q+3}{(q+1)(q+2)}N_g-\frac{2}{q+1}\right)^{\frac{\ln(q+1)}{\ln\left(\frac{(q+1)(q+2)}{2}\right)}}\\
		&+\frac{(q+1)(3q+7)}{2(q+3)}\left(\frac{q+3}{(q+1)(q+2)}N_g-\frac{2}{q+1}\right).
	\end{align*}
	Therefore, for  $g\to\infty$,
	\begin{equation*}
		\lim_{g\to\infty}K_g =\frac{3q+7}{2(q+2)}N_g.
	\end{equation*}
	This finishes the proof.
\end{proof}
Theorem~\ref{Kemeny} shows that for the whole family of networks   $\mathcal{G}_q(g)$, the Kemeny's constant $K_g$ grows as a linear function of $N_g$, the number of nodes, but the factor $(3q+7)/(2(q+2))$ is a decreasing function of $q$.

%\subsection{Kemeny's Constant}

%In addition to the two-node hitting time, the  Kemeny's constant of  $\mathcal{G}_q(g+1)$ can also be expressed in terms of that of $\mathcal{G}_q(g)$.

%\begin{theorem}\label{conKem}
%	For any $g\geqslant0$,
%	\begin{equation*}
	%	\begin{aligned}
		%	K(\mathcal{G}_q(g+1))=&(q+1)K(\mathcal{G}_q(g))+\frac{q(q+1)}{2}M_g\\
		%	&-\frac{q(q+1)}{2(q+2)}N_g.
		%	\end{aligned}
	%	\end{equation*}
%\end{theorem}

%\begin{proof}
%	Suppose that $1=\lambda_1(g)>\lambda_2(g)\geq...\geq\lambda_n(g)\geq-1$ are eigenvalues of the matrix $P$. Since $\mathcal{G}_q(g)$ is a non-bipartite graph. For this case, by Lemma~\ref{lemmaKem} and Theorem~\ref{Spectra}, we have
%	\begin{small}
	%		\begin{equation*}
		%		\begin{aligned}
			%		&K(\mathcal{G}_q(g+1))\\
			%		=   &   \sum_{k=2}^{N_g} \frac{1}{1-\frac{\lambda_k(g)+q}{q+1}}
			%		+\sum_{k=2}^{N_g}
			%		\frac{1}{1+\frac{1}{q+1}}
			%		+\frac{qM_g-N_g}{1-\frac{q-1}{q+1}}\\
			%		+\frac{q(q+1)}{2}M_g-\frac{q(q+1)}{2(q+2)}N_g\\
			%		=   &   (q+1)K(\mathcal{G}_q(g))+\frac{q(q+1)}{2}M_g-\frac{q(q+1)}{2(q+2)}N_g.
			%		\end{aligned}
		%		\end{equation*}
	%	\end{small}

%\end{proof}

\section{Mean Hitting Time}

In this section, we study the mean hitting time for the studied networks with the remarkable scale-free small-world
properties~\cite{WaYiXuZh19}. We will demonstrate that   their mean
hitting time also scales linearly with the number of nodes.

\subsection{Some Definitions}

Here we give  definitions for some quantities related to   network $\mathcal{G}_q(g)$.
\begin{definition}\label{meanhittingtime}
	For network $\mathcal{G}_q(g)$, the mean hitting time is
	\begin{equation}
		\langle{H}_g \rangle= \frac{1}{N_g(N_g-1)}\sum_{i=1}^{N_g}\sum_{j=1}^{N_g}T_{ij}(g).
	\end{equation}
\end{definition}
To obtain the explicit expression of the mean hitting time $\langle{H}_g \rangle$, we first determine three intermediary results for graph $\mathcal{G}_q(g)$, including the sum of hitting times, the additive-degree sum of hitting times, and the multiplicative-degree sum of hitting times.

For network $\mathcal{G}_q(g)$, the sum of hitting times is
\begin{equation}
	H_g=\sum_{i=1}^{N_g}\sum_{j=1}^{N_g}T_{ij}(g);
\end{equation}
the additive-degree sum of hitting times is
\begin{equation}
	H^{+}_g=\sum_{i=1}^{N_g}\sum_{j=1}^{N_g}\left(d_i(g)+d_j(g)\right)T_{ij}(g);
\end{equation}
and the multiplicative-degree sum of hitting times is
\begin{equation}
	H^{\ast}_g=\sum_{i=1}^{N_g}\sum_{j=1}^{N_g}\left(d_i(g)\cdot d_j(g)\right)T_{ij}(g).
\end{equation}

\begin{lemma}\label{conK*}
	For any $g\geqslant0$, the multiplicative-degree hitting time for graph $\mathcal{G}_q(g)$ is
	\begin{small}
		\begin{equation*}
			\begin{aligned}
				H^{\ast}_g=&-\frac{(q+2)(q+4)(q+1)^3}{2}\left(\frac{(q+2)^2(q+1)^3}{4}\right)^g\\
				&+\frac{(q+2)^2(q+1)^3}{q+3}\left(\frac{(q+1)(q+2)}{2}\right)^{2g}\\
				&+\frac{(3q+7)(q+2)^2(q+1)^3}{2(q+3)}\left(\frac{(q+1)(q+2)}{2}\right)^{3g}.
			\end{aligned}
		\end{equation*}
	\end{small}
\end{lemma}
\begin{proof}
	%First notice that $d_i(g+1)=(q+1)d_i$ if $i\in V$, while $d_i(g+1)=2$ if $i\in \mathcal{W}_{g+1}$.
	By definition of the multiplicative-degree sum of hitting times, we have
	\begin{align}\label{Kf+9}
		H^{\ast}_{g+1}	=&\sum_{\{i,j\}\subseteq \mathcal{V}_g\cup \mathcal{W}_{g+1}}\left(d_i(g+1)d_j(g+1)\right)C_{ij}(g+1)\notag\\
		=&4M_g^2K_g.
	\end{align}
	Using Theorem~\ref{Kemeny},  the  result is obtained.
\end{proof}
In what follows, we will determine the other two invariants $H^{+}_g$ and $H_g$ for  network $\mathcal{G}_q(g)$.

\subsection{Some Intermediary Results}

Let $C_{ij}(g)$ be the commute time for any pair of nodes $i$ and $j$ in graph $\mathcal{G}_q(g)$. For any two subsets $X$ and $Y$ of set $\mathcal{V}_g$ of nodes in graph $\mathcal{G}_q(g)$,  define
\begin{equation*}
	C_{X,Y}(g)=\sum_{i\in X, j\in Y}C_{ij}(g).
\end{equation*}
%Then we have the following lemma.
\begin{lemma}\label{xysum}
	For $g\geq0$ and $Y\subseteq\mathcal{V}_g$,
	\begin{equation}
		\sum_{i\in\mathcal{W}_{g+1}}C_{\Gamma(i),Y}(g+1)=\sum_{x\in\mathcal{V}_g}qd_x(g)C_{x,Y}(g+1).
	\end{equation}
\end{lemma}
\begin{proof}
	For any node $x\in\mathcal{V}_g$, there are $d_x(g+1)-d_x(g)=qd_x(g)$ new nodes in $\mathcal{W}_{g+1}$ that are adjacent to $i$. So $C_{x,Y}(g+1)$ is summed $qd_x(g)$ times.
\end{proof}
\begin{lemma}\label{afore1}
	For any $g\geqslant0$,
	\begin{small}
		\begin{equation*}
			\begin{aligned}
				\sum_{i\in \mathcal{W}_{g+1}}\sum_{j\in \mathcal{V}_g}C_{ij}(g+1)
				=&\frac{q(q+1)}{4}H^{+}_q(g)+\frac{q(q+1)}{2}M_g\\
				&\left(3M_gN_g-N_g^2+N_g\right).
			\end{aligned}
		\end{equation*}
	\end{small}
\end{lemma}
\begin{proof}
	By Theorem~\ref{qTriHT}, one obtains
	\begin{align}\label{mr1}
		&\sum_{i\in \mathcal{W}_{g+1}}\sum_{j\in \mathcal{V}_g}C_{ij}(g+1)\notag\\
		=&\sum_{i\in \mathcal{W}_{g+1}}\sum_{j\in \mathcal{V}_g}\left(\frac{3(q+1)}{2}M_g\right.
		\notag\\&\left.\quad\quad\quad\quad\quad\quad+\frac{q+1}{4}\left(2\left(C_{sj}(g)+C_{tj}{g}\right)-C_{st}(g)\right)\right)\notag\\
		=&\frac{3q(q+1)}{2}M_g^2N_g+\frac{q+1}{2}\sum_{i\in \mathcal{W}_{g+1}}\sum_{j\in \mathcal{V}_g}(C_{sj}(g)+C_{tj}(g))\notag\\
		&-\frac{q+1}{4}\sum_{i\in \mathcal{W}_{g+1}}\sum_{j\in \mathcal{V}_g}C_{st}(g).
	\end{align}
	For the second term on the right-hand side of the second equal sign in Eq.~\eqref{mr1}, we have
	\begin{align}\label{mr2}
		&\frac{q+1}{2}\sum_{i\in \mathcal{W}_{g+1}}\sum_{j\in \mathcal{V}_g}(C_{sj}(g)+C_{tj}(g))\notag\\
		=&\frac{q+1}{2}\sum_{i,j\in \mathcal{V}_g}d_i(g)C_{ij}(g)\notag\\
		=&\frac{q+1}{2}\frac{\sum_{i,j\in \mathcal{V}_g}d_i(g)C_{ij}(g)+\sum_{i,j\in \mathcal{V}_g}d_j(g)C_{ij}(g)}{2}\notag\\
		=&\frac{q+1}{4}\sum_{i\in \mathcal{W}_{g+1}}\sum_{j\in \mathcal{V}_g}(d_i(g)+d_j(g))C_{ij}(g)\notag\\
		=&\frac{q+1}{4}H_g^{+}.
	\end{align}
	With respect to the third term in  Eq.~\eqref{mr1},  using Lemma~\ref{Foster}, it can be rewritten as
	\begin{align}\label{mr3}
		\frac{q+1}{4}\sum_{i\in \mathcal{W}_{g+1}}\sum_{j\in \mathcal{V}_g}C_{st}(g)	=&\frac{q(q+1)}{4}N_g\sum_{st\in\mathcal{V}_g}C_{st}(g)\notag\\
		=&\frac{q(q+1)}{2}M_gN_g(N_g-1).
	\end{align}
By plugging Eqs.~\eqref{mr2} and~\eqref{mr3} into Eq.~\eqref{mr1}, we  obtain the desired result.
\end{proof}
\begin{lemma}\label{afore2}
	For any $g\geqslant0$,
	\begin{align*}
		\sum_{i,j\in\mathcal{W}_{g+1}}C_{ij}(g+1)
		=&\frac{q^2(q+1)}{2}H^{\ast}_g+q(q+1)M_g^2\\
		&(3qM_g-qN_g-2).
	\end{align*}
\end{lemma}
\begin{proof}
	Suppose that $\Gamma(i)=\{s,t\}$ and $\Gamma(j)=\{u,v\}$. Note that for any two different nodes $i$ and $j$ in $\mathcal{W}_{g+1}$, if their old  neighbors in $\mathcal{V}_{g}$  are the same, i.e., $\Gamma(i)=\Gamma(j)=\{s,t\}$, we use $i\sim j$ to denote this relation. Otherwise, if the sets of the old neighbors for  $i$ and $j$ are different, we call $i\nsim j$. Then, by Theorem~\ref{qTriHT}, we obtain
	\begin{align}\label{A2_1}
		&   \sum_{i,j\in\mathcal{W}_{g+1}} C_{ij}(g+1)\notag\\
		=&\sum_{i,j\in\mathcal{W}_{g+1}\atop i\neq j,i\sim j}C_{ij}(g+1)+\sum_{i,j\in\mathcal{W}_{g+1}\atop i\nsim j}C_{ij}(g+1)\notag\\
		=&\sum_{i,j\in\mathcal{W}_{g+1}\atop i\nsim j}
		\Bigg(3(q+1)M_g+\frac{q+1}{4}\Big(C_{su}(g)+C_{tu}(g)\notag\\
		&\quad\quad\quad\quad~~+C_{sv}(g)+C_{tv}(g)-\left(C_{uv}(g)+C_{st}(g)\right)\Big)\Bigg)\notag\\
		&+\sum_{i,j\in\mathcal{W}_{g+1}\atop i\neq j, i\sim j}2(q+1)M_{g}\notag\\
		=&3q^2(q+1)M_g^2(M_g-1)+2q(q-1)(q+1)M_g^2\notag\\ &+\frac{q+1}{4}\sum_{i,j\in\mathcal{W}_{g+1}}\left(C_{su}(g)+C_{tu}(g)+C_{sv}(g)+C_{tv}(g)\right)\notag\\
		&-\frac{q+1}{2}\sum_{i,j\in\mathcal{W}_{g+1}\atop i\sim j}C_{st}(g)\notag\\
		&-\frac{q+1}{4}\sum_{i,j\in \mathcal{W}_{g+1}\atop i \nsim j}\left(C_{st}(g)+C_{uv}(g)\right).
	\end{align}
	Below we evaluate the three sum terms on the right-hand side of the second equal sign in Eq.~\eqref{A2_1}. By Lemma~\ref{xysum} and Theorem~\ref{qTriHT},  the first sum term can be computed as
	\begin{small}
		\begin{align}\label{A2_2}
			&  \frac{q+1}{4}\sum_{i,j\in \mathcal{W}_{g+1}}\left(C_{su}(g)+C_{tu}(g)+C_{sv}(g)+C_{tv}(g)\right)\notag\\
			=&\frac{q+1}{4}\sum_{i,j\in\mathcal{W}_{g+1}}C_{\Gamma(i),\Gamma(j)}(g)\notag\\
			=&\frac{q+1}{4}\sum_{x,y\in\mathcal{V}_g}q^2d_x(g)d_y(g)C_{xy}(g)\notag\\
			=&\frac{q^2(q+1)}{2}H^{\ast}_g.
		\end{align}
	\end{small}
	We next compute the second sum term in Eq.~\eqref{A2_1}. By Lemma~\ref{Foster}, we have
	\begin{align}\label{A2_a}
		\frac{q+1}{2}\sum_{i,j\in\mathcal{W}_{g+1}\atop i\sim j}C_{st}(g)=&\frac{q^2(q+1)}{2}\sum_{st\in\mathcal{E}_g}C_{st}(g)\notag\\
		=&q^2(q+1)M_g(N_g-1).
	\end{align}
	We proceed to evaluate the third term in Eq.~\eqref{A2_1}.  According to Eq.~\eqref{div}, it follows that
	\begin{small}
		\begin{align}\label{A2_3}
			&\frac{q+1}{4}\sum_{i,j\in \mathcal{W}_{g+1}\atop i\nsim j}\left(C_{st}(g)+C_{uv}(g)\right)\notag\\=&  \frac{q+1}{4}\sum_{f=1}^{q}\sum_{i\in \mathcal{V}^{(f)}}
			\sum_{i\nsim j} \left(C_{st}(g)+C_{uv}(g)\right)
			\notag\\
			=&  \frac{q+1}{4}q\sum_{st\in \mathcal{E}_g}
			q\sum_{uv\in \mathcal{E}_g \atop uv\neq st}\big(C_{st}(g)+C_{uv}(g)\big)
			\notag\\
			=&  \frac{q^2(q+1)}{4}\sum_{st\in \mathcal{E}_g}
			\sum_{uv\in \mathcal{E}_g\atop uv\neq st}\left(C_{uv}(g)+C_{st}(g)\right)
			.
		\end{align}
	\end{small}
	By Lemma~\ref{Foster}, Eq.~\eqref{A2_3} can be recast as
	\begin{small}
		\begin{align}\label{A2_4}
			&\frac{q+1}{4}\sum_{i,j\in \mathcal{W}_{g+1}}\left(C_{st}(g)+C_{uv}(g)\right)\notag\\
			=&  \frac{q^2(q+1)}{2}(M_g-1)\sum_{st\in \mathcal{E}_g}C_{st}(g)\notag\\
			=&q^2(q+1)M_g(M_g-1)(N_g-1).
		\end{align}
	\end{small}
	Plugging Eqs.~\eqref{A2_2},~\eqref{A2_a}, and~\eqref{A2_4} into Eq.~\eqref{A2_1} gives the result.
\end{proof}

\subsection{Addictive-Degree Sum of Hitting Times}

We now determine the additive-degree hitting time for graph $\mathcal{G}_q(g)$.
\begin{lemma}\label{con6}
	For any $g\geqslant0$, the additive-degree hitting time for graph $\mathcal{G}_q(g)$ is
	\begin{small}
		\begin{equation*}
			\begin{aligned}
				H^{+}_{g}=&\frac{2(q+2)^2(q+1)^3}{(q+3)^2}\left(\frac{(q+1)(q+2)}{2}\right)^{2g}\\	&+\frac{(q\!\!+\!\!2)(3q\!\!+\!\!7)(q\!\!+\!\!1)^3(q^3\!\!+\!\!8q^2\!\!+\!\!22q\!\!+\!\!20)}{(q\!\!+\!\!3)^2(q^2\!\!+\!\!5q\!\!+\!\!8)}\!\left(\!\frac{(q\!+\!1)(q\!+\!2)}{2}\!\right)^{3g}\\
				&-\frac{(q+2)(q+4)(q+1)^3}{(q+3)}\left(\frac{(q+2)^2(q+1)^3}{4}\right)^g\\
				&+\frac{(q+2)(q^2+9q+20)(q+1)^3}{(q+3)(q^2+5q+8)}\left(\frac{(q+2)(q+1)^2}{2}\right)^g\\
				&+\frac{(q+2)(q+1)^3}{(q+3)^2}\left(\frac{(q+1)(q+2)}{2}\right)^g.
			\end{aligned}
		\end{equation*}
	\end{small}
\end{lemma}
\begin{proof}
	%First notice that $d_i(g+1)=(q+1)d_i$ if $i\in V$, while $d_i(g+1)=2$ if $i\in \mathcal{W}_{g+1}$.
	By definition of the additive-degree sum of hitting times, we have
	\begin{small}
		\begin{align}\label{Kf+1}
			H^{+}_{g+1}	=&  \sum_{i,j\in\mathcal{V}_g\cup \mathcal{W}_{g+1}}\left(d_i(g+1)+d_j(g+1)\right)C_{ij}(g+1)\notag\\
			=&\frac{1}{2}\sum_{i,j\in\mathcal{V}_g} \left(d_i(g+1)+d_j(g+1)\right) C_{ij}(g+1)\\
			&+\sum_{i\in \mathcal{W}_{g+1}}\sum_{j\in \mathcal{V}_g}\left(d_i(g+1)+d_j(g+1)\right)C_{ij}(g+1)\notag\\
			&+\frac{1}{2}\sum_{i,j\in \mathcal{W}_{g+1}}\left(d_i(g+1)+d_j(g+1)\right) C_{ij}(g+1).\notag
		\end{align}
	\end{small}
	We begin to compute the three sum terms for $H^{+}_{g+1}$ one by one.
	
	By Theorem~\ref{qTriHT},	the first sum term can be evaluated as
	\begin{align}\label{Kf+2}
		&\frac{1}{2}\sum_{i,j\in \mathcal{V}_g}\left(d_i(g+1)+d_j(g+1)\right)C_{ij}(g+1)\notag\\
		=&\sum_{\{i,j\}\subseteq \mathcal{V}_g} (q+1)\left(d_i(g)+d_j(g)\right)(q+1)C_{ij}(g)\notag\\
		=&(q+1)^2H^{+}_g.
	\end{align}
	For the second sum term, it can be computed as
	\begin{align}\label{Kf+3}
		&    \sum_{i\in \mathcal{W}_{g+1}}\sum_{j\in \mathcal{V}_g}\left(d_i(g+1)+d_j(g+1)\right)C_{ij}(g+1)\notag\\
		=&\sum_{i\in \mathcal{W}_{g+1}}\sum_{j\in \mathcal{V}_g}\left((q+1)+(q+1)d_j(g)\right)C_{ij}(g+1)\notag\\
		=&  (q+1)\sum_{i\in \mathcal{W}_{g+1}}\sum_{j\in \mathcal{V}_g}C_{ij}(g+1)\\
		&+(q+1)\sum_{i\in \mathcal{W}_{g+1}}\sum_{j\in\mathcal{V}_g}d_j(g)C_{ij}(g+1)\notag,
	\end{align}
	where the two sum terms can be further computed as follows.
	First, by Lemma~\ref{afore1},
	\begin{small}
		\begin{align}\label{Kf+4}
			&(q+1)\sum_{i\in \mathcal{W}_{g+1}}\sum_{j\in \mathcal{V}_g}C_{ij}(g+1)\notag\\
			=&(q+1)\bigg(\frac{q(q+1)}{4}H^{+}_g+\frac{q(q+1)}{2}M_g\notag\\
			&\quad\quad\quad~~\left(3M_gN_g-N_g^2+N_g\right)\bigg)\\
			=&\frac{q(q+1)^2}{4}H^{+}_g+\frac{q(q+1)^2}{2}M_g\left(3M_gN_g-N_g^2+N_g\right)\notag.
		\end{align}
	\end{small}
	On the other hand, by Lemma~\ref{Foster} and Theorem~\ref{qTriHT},
	\begin{align}\label{Kf+5}
		&     (q+1)\sum_{i\in \mathcal{W}_{g+1}}\sum_{j\in \mathcal{V}_g}d_j(g)C_{ij}(g+1)\notag\\
		=&     (q+1)^2\sum_{i\in \mathcal{W}_{g+1}}\sum_{j\in \mathcal{V}_g}d_j(g)
		\bigg(\frac{3}{2}M_g+\frac{1}{4}(2(C_{sj}(g)\notag\\
		&\quad\quad\quad\quad\quad\quad\quad\quad\quad\quad\quad~~+C_{tj}(g))-C_{st}(g))\bigg)\notag\\
		=&(q+1)^2\cdot qM_g\cdot2M_g\cdot\frac{3}{2}M_g\notag\\
		&+\frac{(q+1)^2}{2}\sum_{i\in \mathcal{W}_{g+1}}\sum_{j\in \mathcal{V}_g}d_j(g)(C_{js}(g)+C_{tj}(g))\notag\\
		&-\frac{(q+1)^2}{4}\sum_{j\in \mathcal{V}_g}d_j(g)C_{st}(g)\\
		=&3q(q+1)^2M_g^3\notag\\
		&+\frac{(q+1)^2}{2}\sum_{i\in \mathcal{W}_{g+1}}\sum_{j\in \mathcal{V}_g}d_j(g)(C_{js}(g)+C_{tj}(g))\notag\\
		&-\frac{(q+1)^2}{4}2M_g\cdot 2M_g(N_g-1)\notag\\
		=&3q(q+1)^2M_g^3\notag\\
		&+\frac{(q+1)^2}{2}\sum_{i\in \mathcal{W}_{g+1}}\sum_{j\in \mathcal{V}_g}d_j(g)(C_{js}(g)+C_{tj}(g))\notag\\
		&-q(q+1)^2M_g^2(N_g-1),\notag
	\end{align}
	while  the middle part can be computed to obtain
	\begin{align}\label{Kf+6}
		&   \frac{(q+1)^2}{2}\sum_{i\in \mathcal{W}_{g+1}}\sum_{j\in \mathcal{V}_g}d_j(g)\left(C_{js}(g)+C_{tj}(g)\right)\notag\\
		=&   \frac{(q+1)^2}{2}q\sum_{i\in \mathcal{V}^{(1)}}\sum_{j\in \mathcal{V}_g}d_j(g)\left(C_{js}(g)+C_{tj}(g)\right)\notag\\
		=&  \frac{q(q+1)^2}{2}\sum_{j\in\mathcal{V}_g}\sum_{i\in \mathcal{V}^{(1)}}d_j(g)\left(C_{js}(g)+C_{tj}(g)\right)\notag\\
		= &  \frac{q(q+1)^2}{2}\sum_{j\in \mathcal{V}_g}\sum_{k\in \mathcal{V}_g}d_j(g) d_k(g) C_{kj}(g)\notag\\
		=&  q(q+1)^2H^{\ast}_g.
	\end{align}
	Combining Eqs.~\eqref{Kf+3}-\eqref{Kf+6} yields
	\begin{align}\label{Kf+7}
		&   \sum_{i\in \mathcal{W}_{g+1}}\sum_{j\in \mathcal{V}_g}(d_i(g+1)+d_j(g+1)) C_{ij}(g+1)\notag\\
		=&\frac{q(q+1)^2}{4}H^{+}_g+q(q+1)^2H^{\ast}_g\\
		&+\frac{q(q+1)^2}{2}M_g\left(6M_g^2+M_gN_g-N_g^2+2M_g+N_g\right)\notag.
	\end{align}
	
	With regard to the third sum term in Eq.~\eqref{Kf+1}, by Lemma~\ref{afore2}, we have
	\begin{align}\label{Kf+8}
		&\sum_{i,j\in\mathcal{W}_{g+1}} (d_i(g+1)+d_j(g+1)) C_{ij}(g+1)\notag\\
		=&  2(q+1)\bigg(\frac{q^2(q\!+\!1)}{2}H^{\ast}_g+q(q+1)M_g^2\notag\\
		&\quad\quad\quad\quad~(3qM_g-qN_g-2)\bigg)\\
		=&  q^2(q+1)^2H^{\ast}_g+2q(q+1)^2M_g^2\left(3qM_g-qN_g-2\right)\notag.
	\end{align}
	Substituting Eqs.~\eqref{Kf+2}, \eqref{Kf+7} and~\eqref{Kf+8} back into Eq.~\eqref{Kf+1} gives
	\begin{equation*}
		\begin{aligned}
			H^{+}_{g+1}	=&\frac{(q+2)(q+1)^2}{2}H^{+}_g+\frac{q(q+2)(q+1)^2}{2}H^{\ast}_g\\
			&+\frac{1}{2}q(q+1)^2M_g(N_g+2M_g)\left(3M_g-N_g+1\right)\notag\\
			&+q(q+1)^2M_g^2\left(3qM_g-qN_g-2\right).
		\end{aligned}
	\end{equation*}
	Considering the initial condition $H^{+}_0=2(q+2)(q+1)^3$, the above recursive relation is solved to yield  the deriable result.
\end{proof}

\subsection{Mean Hitting Time}

%We finally determine mean hitting time for $\mathcal{G}_q(g)$.

We are now ready to present the result for mean hitting time
of $\mathcal{G}_q(g)$, denoted by $\langle{H}_g \rangle$, and its dominant behavior.
\begin{theorem}\label{mht}
	For any $g\geq0$, the mean hitting time  for graph $\mathcal{G}_q(g)$ is
	\begin{small}
		\begin{align}\label{m1}
			&\langle{H}_g \rangle=\notag\\
			&\frac{(q\!+\!3)^2}{(q\!+\!1)^2(q\!+\!2)^2\!\left(\left(\frac{(q\!+\!1)(q\!+\!2)}{2}\right)^g\!\!\!+\!\frac{2}{q\!+\!1}\right)\!\!\left(\left(\frac{(q\!+\!1)(q\!+\!2)}{2}\right)^g\!\!\!+\!\frac{1}{q\!+\!2}\right)}\notag\\
			&\bigg(\frac{(q+1)(q+2)^2(q^3+8q^2+15q+8)}{(q+3)^2(q^2+5q+8)}\left(\frac{(q+1)(q+2)}{2}\right)^{2g}\notag\\
			&+\frac{(q+4)(3q+7)(q+2)^2(q+1)^3}{2(q+3)^2(q^2+5q+8)}\left(\frac{(q+1)(q+2)}{2}\right)^{3g}\notag\\
			&-\frac{(q+2)(q+4)(q+1)^3}{2(q+3)^2}\left(\frac{(q+2)^2(q+1)^3}{4}\right)^{g}\notag\\
			&+\frac{(q+2)(q^2+9q+20)(q+1)^3}{(q+3)^2(q^2+5q+8)}\left(\frac{(q+2)(q+1)^2}{2}\right)^g\notag\\
			&+\frac{2(q+2)(q+1)^2(q+4)^2}{(q+3)^2(q^2+5q+8)}(q+1)^g\notag\\
			&-\frac{(q+2)(q+1)^2}{(q+3)^2}\left(\frac{(q+1)(q+2)}{2}\right)^g\bigg).
		\end{align}
	\end{small}
	When $g\to\infty$,
	\begin{small}
		\begin{equation}\label{m1x}
			\lim_{g\to\infty}\langle{H}_g \rangle =\frac{(q+3)(q+4)(3q+7)}{2(q+2)(q^2+5q+8)}N_g.
		\end{equation}
	\end{small}
\end{theorem}
\begin{proof}
	Since $\langle{H}_g \rangle=H_{g}/(N_{g}(N_{g}-1))$, in order to determine $\langle{H}_g \rangle$, we first determine $H_{g}$. 	For network $\mathcal{G}_q(g+1)$, we have
	\begin{align}\label{Kf1}
		H_{g+1}=&\frac{1}{2}\sum_{i,j\in\mathcal{V}_g\cup \mathcal{W}_{g+1}} C_{ij}(g+1)\notag\\
		=&\frac{1}{2}\sum_{i,j\in\mathcal{V}_g}C_{ij}(g\!+\!1)\!
		+\!\sum_{i\in \mathcal{W}_{g+1}}\!\sum_{j\in \mathcal{V}_g}\!C_{ij}(g\!+\!1)\notag\\
		&+\frac{1}{2}\sum_{i,j\mathcal{W}_{g+1}} C_{ij}(g+1).
	\end{align}
	Below we will compute the three sum terms in Eq.~\eqref{Kf1}. By Theorem~\ref{qTriHT},
	the first sum term can be evaluated as
	\begin{align}\label{Kf2}
		\frac{1}{2}\sum_{i,j\in \mathcal{V}_g}C_{ij}(g+1)=&\sum_{\{i,j\}\subseteq \mathcal{V}_g}(q+1)C_{ij}(g)\notag\\
		=&(q+1)H_g.
	\end{align}
	Using Lemma~\ref{afore1},	 the second sum term is determined as
	\begin{equation}\label{Kf9}
		\begin{aligned}
			\sum_{i\in \mathcal{W}_{g+1}}\sum_{j\in \mathcal{V}_g}C_{ij}(g+1)
			=&\frac{q(q+1)}{4}H^{+}_q(g)+\frac{q(q+1)}{2}M_g\\
			&\left(3M_gN_g-N_g^2+N_g\right).
		\end{aligned}
	\end{equation}
	Finally, by Lemma~\ref{afore2},  the third sum term is computed as
	\begin{align}\label{Kf10}
		&\sum_{i,j\in\mathcal{W}_{g+1}} C_{ij}(g+1)\notag\\
		=&\frac{q^2(q\!+\!1)}{2}H^{\ast}_g+q(q+1)M_g^2(3qM_g-qN_g-2).
	\end{align}
	Plugging Eqs.~\eqref{Kf2}-\eqref{Kf10} into Eq.~\eqref{Kf1} leads to
	\begin{small}
		\begin{equation*}
			\begin{aligned}
				H_{g+1}	=&(q+1)H_g\!+\!\frac{q(q+1)}{2}H^{+}_q(g)\!+\!\frac{q^2(q+1)}{4}H^{\ast}_q(g)\\
				&+\frac{1}{2}q(q+1)M_gN_g(3M_g-N_g+1)\\
				&+\frac{1}{2}q(q+1)M_g^2(3qM_g-qN_g-2).
			\end{aligned}
		\end{equation*}
	\end{small}
	Considering the initial condition $H_0=(q+2)(q+1)^2$, the recursive relation is solved to obtain
	\begin{small}
		\begin{equation*}
			\begin{aligned}
				H_g=&\frac{(q+1)(q+2)^2(q^3+8q^2+15q+8)}{(q+3)^2(q^2+5q+8)}\left(\frac{(q+1)(q+2)}{2}\right)^{2g}\\
				&+\frac{(q+4)(3q+7)(q+2)^2(q+1)^3}{2(q+3)^2(q^2+5q+8)}\left(\frac{(q+1)(q+2)}{2}\right)^{3g}\\
				&-\frac{(q+2)(q+4)(q+1)^3}{2(q+3)^2}\left(\frac{(q+2)^2(q+1)^3}{4}\right)^{g}\\
				&+\frac{(q+2)(q^2+9q+20)(q+1)^3}{(q+3)^2(q^2+5q+8)}\left(\frac{(q+2)(q+1)^2}{2}\right)^g\\
				&+\frac{2(q+2)(q+1)^2(q+4)^2}{(q+3)^2(q^2+5q+8)}(q+1)^g\\
				&-\frac{(q+2)(q+1)^2}{(q+3)^2}\left(\frac{(q+1)(q+2)}{2}\right)^g.
			\end{aligned}
		\end{equation*}
	\end{small}
	Plugging this result to $\langle{H}_g \rangle=H_{g}/(N_{g}(N_{g}-1))$ gives~\eqref{m1}.
	
	In a similar way to that of Kemeny's constant $K_{g}$, we can represent mean hitting time $\langle{H}_g \rangle$  in terms of
	the number of nodes $N_g$, and obtain the leading term of $\langle{H}_g \rangle$ given by~\eqref{m1x}.
\end{proof}
Theorem~\ref{mht} indicates that  mean hitting time $\langle{H}_g \rangle$  of network  $\mathcal{G}_q(g)$ scales  linearly as $N_g$ with the factor decreasing with $q$, which is similar to that for the Kemeny's constant $K_g$.

% or


\section{Conclusion}

The edge corona product of a graph is a natural extension of traditional triangulation operation, which has been successfully applied to generate complex networks with prominent properties observed in various real-life systems. In this paper, we presented an extensive study of various properties for hitting times of random walks on a class of graphs, which are iteratively generated by edge corona product of complete graphs. We first deduced recursive formulas for the eigenvalues and eigenvectors of normalized adjacency matrix of the graphs under consideration. Using these results, we then determined a recursive expression for two-node hitting time from an arbitrary node to another. Also, we obtained exact solution to the Kemeny's constant, which is a weighted average of hitting times among all node pairs. Finally, we provided analytical formulas for the sum of hitting times, the sum of multiplicative-degree hitting times, and the sum of additive-degree hitting times.

%The edge corona product of a graph is a natural extension of traditional triangulation operation, which has been successfully applied to generate complex networks with prominent properties observed in various real-life systems. In this paper, we presented an extensive study of various properties for hitting times of random walks on a class of graphs, which are iteratively generated by edge corona product of complete graphs. This graph family is a skeleton of simplicial complexes and can be considered as projections of hypergraphs, thus can capture group interactions between three or more nodes at a time. Both simplicial complexes and hypergraphs have been the focus of research from the scientific community.

%We first deduced recursive formulas for the eigenvalues and eigenvectors of normalized adjacency matrix of the graphs under consideration. Using these results, we then determined a recursive expression for two-node hitting time from an arbitrary node to another. Also, we obtained exact solution to the Kemeny's constant, which is a weighted average of hitting times among all node pairs. Finally, we provided analytical formulas for the sum of hitting times, the sum of multiplicative-degree hitting times, and the sum of additive-degree hitting times. Our work provides deep insights of hitting times for random walks on networks with higher-order interactions.
\section*{Acknowledgements}

This work was supported by the National
Natural Science Foundation of China (Nos. 61872093, U20B2051, 62272107 and U19A2066),  the Shanghai Municipal Science and Technology Major Project (No.2021SHZDZX0103), the Innovation Action Plan of Shanghai Science and Technology (No. 21511102200), the Key R \& D Program of Guangdong Province (No. 2020B0101090001), and Ji Hua Laboratory, Foshan, China (No.X190011TB190). Mingzhe Zhu was also supported by Fudan's Undergraduate Research
Opportunities Program (FDUROP) under Grant No. 20001.

\section*{Data Availability Statement}

No new data were generated or analysed in support of this research.

\bibliographystyle{compj}
\bibliography{Triangulation}

\end{document}